\documentclass[11pt]{article}

\def\Rp{\Re^+\cup\{0\}}

\def\R{\mathbb{R}}

\def\Rp{\R^+\cup\{0\}}

\usepackage{amsmath}
\usepackage{amsthm}
\usepackage{amssymb}
\usepackage{xspace}
\usepackage[pdftex]{graphicx}  
\usepackage{typearea}
\usepackage{rotating}
\typearea{14}

\newtheorem{theorem}{Theorem}

\newtheorem{lemma}[theorem]{Lemma}

\newcommand{\junk}[1]{}

\newcommand{\keywords}[1]{\medskip\par\noindent\textbf{Keywords: }#1}

\title{On the Complexity of Nash Dynamics and Sink equilibria}

\author{
{Vahab S. Mirrokni}
\thanks{Theory Group, Microsoft Research,
E-Mail: mirrokni@microsoft.com. }
 \and
{Alexander Skopalik}
\thanks{Dept.\ of Computer Science, RWTH Aachen,
E-Mail: skopalik@cs.rwth-aachen.de.  Research supported in part
by the German-Israeli Foundation.}}

\date{}

\newcommand{\ignore}[1]{}

\begin{document}

\maketitle

\maketitle
\thispagestyle{empty}

\quad\bigskip

\begin{abstract}
Studying Nash dynamics is an important approach for analyzing the outcome of games with 
repeated selfish behavior of self-interested agents.
Sink equilibria has been introduced by Goemans, Mirrokni, and Vetta 
for studying social cost on Nash dynamics 
over pure strategies in games. However, 
they do not address the complexity of  sink equilibria in these games.  
Recently,  Fabrikant and Papadimitriou initiated the study of 
the complexity of Nash dynamics in two classes of games.
In order to completely understand
the complexity of Nash dynamics in a variety of games, 
we study the following three questions for various games:
(i) given a state in game, can we verify if this state is in a sink equilibrium
or not?
(ii) given an instance of a
game, can we verify if there exists any sink equilibrium other than pure 
Nash equilibria? and 
(iii) given an instance of a game, can we verify if there exists
a pure Nash equilibrium (i.e, a sink equilibrium with one state)? 

In this paper, we almost answer all of the above questions for a variety of  
classes of games with succinct representation, including  
anonymous games, player-specific and weighted congestion games,
valid-utility games, and two-sided market games.
In particular,  for most of these problems, we show that
(i) it is PSPACE-complete to verify if a given state is in a 
sink equilibrium, 
(ii) it is NP-hard to verify if there exists a pure Nash equilibrium 
in the game or not,
(iii) it is PSPACE-complete to verify if there exists
any sink equilibrium other than pure Nash equilibria.
To solve these problems, we illustrate general techniques that  
could be used to answer similar questions in other 
classes of games. 

\end{abstract}

\keywords{Nash equilibria, potential games, sink equilibria}
\vfill

\newpage
\setcounter{page}{1}

\section{Introduction}
A standard approach in studying the outcome of a system involving 
self-interested behavior of agents is to investigate the 
Nash dynamics of the corresponding games. In 
Nash dynamics, agents repeatedly respond to the current state 
of the game by playing a best-response strategy. Studying such dynamics
is very important for understanding the behavior of a system 
throughout time, 
and the outcome of the game after many repeated game play. 
Similar to the recent efforts in studying the complexity of game 
theoretic concepts such as mixed Nash equilibria~\cite{DGP06,CD06}, 
and pure NE~\cite{FPT04,SV07}, studying the complexity of Nash dynamics 
can help us better understand the outcome of a game.

In an attempt to study such dynamics for 
pure strategies, Goemans, Mirrokni,
and Vetta~\cite{GMV05} introduced the concept of sink equilibria in games: 
sink equilibria are strongly connected components of a strategy profile graph  
associated with the game with no outgoing 
edges. Equivalently, 
sink equilibria characterize all states for which the
probability of reaching that state after  a sufficiently large random 
best-response sequence is nonzero. Also any random best-response sequence
will converge to a sink equilibrium with probability one. Moreover,
sink equilibria generalize pure Nash equilibria in that a 
pure Nash equilibrium is a single-state sink equilibrium of the game.


Goemans et al.~\cite{GMV05} studied sink equilibria for their social cost
in two classes of games. However, they did not consider 
the complexity of sink equilibria or Nash dynamics in those games.
Recently, Fabrikant and Papadimitriou~\cite{FP08} initiated
the study of the complexity of sink equilibria. by studying
 the problem of verifying if a state is in a sink equilibria
for two classes of games.
Extending on these ideas, we formalize several questions related to Nash dynamics
of various games and completely
study the complexity of the Nash dynamics and sink equilibria in these games.

Sink equilibria characterize all strategy profiles in the game
with a nonzero probability of reaching them after a long enough best-response
walk. Therefore, given a strategy profile, in order to verify if there is a non-zero
probability of reaching this state after a sufficiently long random best-response walk
we need to verify if this state is in a sink equilibrium or not.
This problem has been considered by Fabrikant and Papadimitriou~\cite{FP08}
for two classes of games, and is as follows:

{\noindent \bf {\sc In a Sink} problem.} Given an instance of 
a game and a strategy profile in this game, can we verify if this strategy profile belongs
to any sink equilibria or not?

For a given state in a game, an interesting problem is to estimate
the probability of reaching this state after a long random best-response walk.
Note that a hardness result for {\sc in a sink} problem implies that 
for a given state, even approximating this probability is a computationally hard 
problem, (since distinguishing the probability of zero and nonzero is hard).
Fabrikant and Papadimitriou showed that {\sc in a sink} problem is PSPACE-hard 
for graphical games and a BGP next-hop routing game~\cite{FP08}.
We show that this problem is PSPACE-complete for weighted/player-specific congestion
games, valid-utility games, two-sided market games, and anonymous games.
The proofs for all the above games except anonymous games 
are similar and based on a reduction from halting problem of
a space bounded Turing machine. The proof for anonymous games has
unique features and is different from the rest.

Given an instance of a game,
it is very helpful to know if the random repeated self-interested
actions of the agents in the game can cycle forever 
or such dynamics will converge to a pure Nash equilibria with
probability one. This problem is related to characterizing the
structure of sink equilibria in a game, and 
in particular the existence of {\em non-singleton sink equilibria}.
Having such a sink equilibrium
indicates that even random Nash dynamics may also converge 
to an everlasting cycle.
As a result, we formalize the following problem in games:

{\noindent \bf {\sc Has a Non-singleton Sink} problem.} Given an instance of 
a game, can we verify if this game possesses a non-singleton sink equilibrium, i.e., 
sink equilibria other than pure Nash equilibria. 

Pure Nash equilibria (if they exist) are local optima of the Nash dynamics.
Other than the problem of computing a pure Nash equilibrium in various games, 
the problem of verifying if such equilibria exist has been studied 
for various classes of games. 
We complement the previous questions with the following problem:

{\noindent \bf {\sc Has a Singleton Sink} problem.} Given an instance of 
a game, can we verify if this game possesses a pure Nash 
equilibrium (singleton sink equilibrium)?
%
%

Answering all the above questions for a game gives a thorough understanding
of the complexity of Nash dynamics and the complexity of characterizing sink equilibria 
in that game.

{\bf\noindent Our Results.}
We study the above four problems in a variety of games with succinct representation 
including  player-specific and weighted congestion games,  
anonymous games, valid-utility games, and two-sided market games.
All of these games are well-studied for their
existence of pure Nash equilibria, complexity of mixed and pure NE, 
or/and their price of anarchy for different social functions~\cite{GS62,R73,M96,JL07,DP07}.
To solve these problems, we illustrate general techniques that  
could be used as tools to answer similar questions for other classes of games.

Fabrikant and Papadimitriou showed that {\sc in a sink} problem is PSPACE-hard 
for graphical games and a BGP next-hop routing game~\cite{FP08}. They posed
this problem as an open question for weighted congestion games, and valid-utility
games.
We show that this problem is PSPACE-complete for weighted/player-specific congestion
games, valid-utility games, two-sided market games, and anonymous games.
The proofs for all the above games except anonymous games 
are similar and based on a reduction from halting problem of
a space bounded Turing machine. The proof for anonymous games has
unique features and is different from the rest. The hardness of the
{\sc in a sink} problem in anonymous games is despite the fact that 
approximate pure Nash equilibria can be computed in these games in
polynomial time~\cite{DP}. 


For {\sc Has a non-singleton sink} problem, we prove that
it is PSPACE-complete for weighted/player-specific congestion
games, valid-utility games, two-sided market games, 
and anonymous games.
The reductions for {\sc Has a non-singleton sink} problem extend the proofs for the  {\sc in a sink} problem.

{\sc Has a singleton sink} problem has been well-studied for all games in this
paper except for valid-utility games and two-sided market games. We show that
{\sc has a singleton sink problem} is NP-hard for these games as well.
Our results for two-sided markets characterize the complexity of existence of 
a stable matching in many-to-one two-sided matching markets; an extensively studied
problem in the economics literature~\cite{GS62,RV90,KU06}. Existing results for 
many-to-one two-sided markets give sufficient conditions for existence of 
stable matchings (or pure Nash equilibria) in different variants of the problem~\cite{GS62,RV90,KU06}, but 
they have not explored the complexity of verifying the existence of stable matchings
(or pure Nash equilibria) in these games.

{\bf \noindent Related Work.}
Prior to this paper, the {\sc Has a non-singleton sink}
problem has not been studied for any of the above games.
{\sc In a sink} problem has been studied only for 
graphical games~\cite{FP08}.
{\sc Has singleton Sink} problem, however, has been studied
extensively for all the above games except valid-utility games
and two-sided market games. In fact, it has been shown
that {\sc has a singleton sink} problem is {\sf NP}-hard
for weighted congestion games and local-effect games\cite{DS06}, 
player-specific congestion games~\cite{AS07}, graphical games~\cite{FP08},
and action-graph games~\cite{JL07}.
For anonymous games it has been shown that 
hat an approximate NE  are computable in 
polynomial time\cite{DP07b} and that {\sc has a singleton sink} is {\sf TC$^0$}-complete\cite{BFH07}.

There has been a recent significant progress 
in understanding the complexity of equilibria in games.
The complexity of mixed Nash equilibria is now 
well-understood by the recent results on PPAD-hard-ness of computing
mixed NE\cite{DGP06,CD06}, and even for computing
approximate mixed NE\cite{CDT06}.
The complexity of pure Nash equilibria
in various games (especially congestion games) have also
been well-studied by recent results
on PLS-completeness of computing a pure 
Nash equilibrium\cite{FPT04,ARV06a}, and 
even for computing an approximate pure NE~\cite{SV07}.

\section{Preliminaries}
\subsection{General Definitions}
{\noindent \bf Strategic games.} A strategic game (or a
normal-form game) $\Lambda=<N,(\Sigma_i),(u_i)>$ has a finite set
$N=\{1, \ldots , n\}$ of players. Player $i\in N$ has a set
$\Sigma_i$ of strategies (or strategies). The whole strategy set is
$\Sigma = \Sigma_1 \times \cdots \times \Sigma_n$ and a 
{\em strategy profile} $S\in \Sigma$ is also called a \emph{profile} or {\em state}. 
The utility function of player $i$ is $u_i:
\Sigma \rightarrow \R$, which maps the joint strategy $S\in\Sigma$
to a real number. Let $S=(s_1,\ldots,s_n)$ denote the profile of
strategies taken by the players, and let
$s_{-i}=(s_1,\ldots,s_{i-1},s_{i+1},\ldots,s_n)$ denote the
profile of strategies taken by all players other than player $i$.
Note that $S=(s_i,s_{-i})$. An {\em improvement} move $s'_i$ for a
player $i$ in a  profile $S$ is a move for which $u_i(s_{-i},
s'_i) \ge u_i(S)$. A {\em best response} move $S''_i$ for a player
$i$ in a profile $S$ is an improvement move that has the maximum
utility. 
Note that in cost
minimizing games, each player $i$ wants to minimize the cost
$c_i(S) = -u_i(S)$ in strategy profile $S$. This type of games
include  congestion games with delay functions on edges which 
will be defined later.

\noindent{\bf Nash equilibria (NE):}
A strategy profile $S \in \Sigma$ is a \emph{pure Nash equilibrium} if
no player $i\in N$ can benefit from unilaterally deviating from
his strategy to another strategy, i.e., $\forall i\in N \;\forall
s_i'\in \Sigma_i \;:\; u_i(s_{-i},s_i')\leq u_i(S)$. 
We can also define $\alpha$-Nash equilibria as follows.
For $1>\alpha>0$, a state $S$ is an $\alpha$-Nash equilibrium if
for every player $i$, $c_i(s_{-i},s'_i) \geq (1-\alpha)c_i(S)$ for
all $s'_i \in \Sigma_i$.

{\noindent \bf State graph.} Given any game $\Lambda$, the state
graph $G(\Lambda)$ is an arc-labeled directed graph as follows.
Each vertex in the graph represents a joint strategy $S$. There is
an arc from state $S$ to state $S'$ with label $i$ iff there
exists player $i$ and strategy $s'_i \in \Sigma_{i}$ such that
$S'=(s_{-i},s'_i)$, i.e., $S'$ is obtained from $S$ by a move of a
single player $i$ that improves his utility from $S$ to $S'$.

{\noindent \bf Nash dynamics.}
A Nash dynamics or best-response dynamics is equivalent to a walk
in the state graph.

{\noindent \bf Sink equilibria.} Given any game $\Lambda$, 
a sink equilibrium is a subset of states $T$ that form a
strongly connected component of the state 
graph such that there is no outgoing edge from states in $T$
to any state outside $T$. As a result, any pure Nash equilibrium
of a game is a single-state sink equilibrium, and 
a game may have several sink equilibria.

\subsection{Definition of games}
{\noindent \bf {(Unweighted) Congestion Games.}}
An (unweighted)
congestion game is defined by
a tuple $<N,E,(\Sigma_i)_{i \in N},(d_e)_{e \in E}>$ where
$E$ is a set of resources, $\Sigma_i \subseteq 2^E$ is the strategy
space of player $i$, and $d_e:\mathbb{N}\rightarrow\mathbb{Z}$ is a
delay function associated with resource $e$. For a strategy profile 
$S = (s_1, \ldots, s_n)$,
we define the \emph{congestion} $n_e(S)$ on resource $e$ by
$n_e(S)=|\{i | e \in s_i\}|$, that is $n_e(S)$ is the number of
players that selected an strategy containing resource $e$ in $S$. The
cost (or delay) $c_i(S)$ of player $i$ in a strategy profile $S$ is
$c_i(S)=-u_i(S)=\sum_{e\in s_i} d_e(n_e(S))$. 

In {\em weighted congestion games}, 
player $i$ has weighted demand $w_i$.
In this game, the congestion (load) on resource $e$ in a state $S$,
denoted by  by $l_e(S)$ is as follows $l_e(S)=\sum_{i | e \in s_i} w_i$. 
The cost or delay of players is defined the same way as the congestion games.
A {\em player-specific
congestion game} is defined by
a tuple $<N,E,(\Sigma_i)_{i \in N},(d_{e,i})_{e \in E,i\in N}>$ where
$E$ and $\Sigma_i \subseteq 2^E$ are the same 
as congestion games, and $d_{e,i}:\mathbb{N}\rightarrow\mathbb{Z}$ is a
delay function associated with resource $e$ and player $i$.
The congestion $n_e(S)$ on resource $e$ is
defined the same as congestion games. The
cost (delay) $c_i(S)$ of player $i$ in a strategy profile $S$ is
$c_i(S)=-u_i(S)=\sum_{e\in s_i} d_{e,i}(n_e(S))$. 
%

{\bf \noindent Many-to-one Two-sided Markets.} We model the 
{\em many-to-one} two-sided market  $({\cal X}, {\cal 
Y})$ between two sides of active agents $\cal X$ and passive agents $\cal Y$ 
as a game $G({\cal X}, {\cal Y})$ among active agents $x\in {\cal 
X}$. The strategy set of each active agent $x\in {\cal X}$ is a 
lower-ideal~\footnote{A family $\cal F$ of subsets is lower-ideal if and only if 
for any subset $S\in {\cal F}$ and $S'\in S$, then $S'\in {\cal F}$.} 
family of subsets of passive agents ${\cal F}_x$ where ${\cal F}_x\subseteq 
2^{\cal Y}$, i.\,e., an active agent $x\in {\cal X}$ can play a subset 
$s_x\in {\cal F}_x$ of passive agents. Each agent $x\in {\cal X}$
also has a preference (a.k.a social choice) over its strategies.
This preference is capture by a utility function 
$u_x: 2^{\cal Y}\rightarrow \R$ which assigns a utility, $u_x(T)$, 
to each  subset $T\subseteq \cal Y$.
Each agent $y\in {\cal Y}$ has a 
strict preference list over the set of agents $x\in {\cal X}$ that can 
play this set, i.\,e., $x$ is preferred to $x'$ by $y$ iff $u_y(x) > 
u_y(x')$. We assume that $u_y(x) \neq u_y(x')$ for any two agents $x$ and 
$x'$. Given a vector of strategies ${\cal S}=(s_1, \ldots, s_n)$ for active 
agents, agent $y$ is {\em matched} to the best agent $x\in {\cal X}$ in 
the preference list of agent $y$ such that $y\in s_x$.  In this case, 
we say that $x$ is the {\em winner} of agent $y$, or equivalently, 
agent $x$ {\em wins} agent $y$. The goal of each active agent 
$x$ is to maximize the utility of the set of passive agents that she wins.
Given a strategy profile $S$, let $T_x(S)\subseteq s_x$ be the 
set of passive agents that agent $x$ wins.
The utility of player $x$ in strategy profile $S$ is equal to 
$u_x(T_x(S))$, the goal of $x$ is to maximize this utility. 

It is not see that pure Nash equilibria of the above game 
correspond to stable matchings for many-to-one two-sided markets as defined by ...

{\bf \noindent Valid-utility Games.}
Here we briefly define the class of valid-utility games; see \cite{V02}
for more details. 
In valid-utility games, for each player $i$, there exists
a ground set of markets $V_i$. We denote by ${V}$
the union of ground sets of all players, i.e., 
${V}= \cup_{i\in U} V_i$.
The feasible strategy set $F_i$ of player $i$ is a subset
of the power set, $2^{V_i}$, of $V_i$. Thus, 
a strategy $s_i$ of player $i$ is a subset of $V_i$ 
($s_i\subseteq V_i$). The empty set, denoted $\emptyset_i$ for player $i$, 
corresponds to player $i$ taking no action.

Let ${\cal G} (U,\{F_i\vert i\in U\}, \{u_i() \vert i\in U\})$ be a
non-cooperative 
strategic game where $F_i\subseteq 2^{V_i}$ is a family 
of feasible strategies for player $i$. 
Let $V= \cup_{i\in U} V_i$ and 
let the social function be 
$\gamma: \Pi_{i\in U} 2^V \rightarrow \Rp$. 
Then $\cal G$ is  a {\em valid-utility game} if it satisfies 
the following properties: (1)  The social function $\gamma$ is
submodular and non-decreasing, (2) The utility of a player is at least
the difference in the
social  function when the player participates 
versus when it does not 
participate.
and (3) For any strategy profile, the sum of the utilities of players
should be less than or equal to the social function for that
strategy profile.

This framework encompasses a wide range of games including the
facility location games, traffic routing games, auctions~\cite{V02},
market sharing games~\cite{GLMT04}, 
and distributed 
caching games~\cite{FGMS06}.
In \cite{V02} it was shown that
the price of anarchy (for mixed Nash equilibria) 
in valid-utility games is at most $2$.
%

{\bf \noindent Anonymous games.} 
Anonymous game\cite{DP07} are games in which players have
the {\em same strategy sets}, but different utilities for 
the same strategies; however, these
utilities do not depend on the identity of the other players, but
only on the {\em number} of other players taking each action.
An interesting subclass of these games is anonymous 
games with a {\em constant-size strategy set} in which
the size of the strategy set of players is a fixed constant.

\section{Existence of Pure Nash Equilibria}
In this section, we study the {\sc Has a Pure} problem for succinct games. This
problem has been already considered and resolved for weighted congestion games~\cite{} and
player-specific congestion games~\cite{}. We resolve this problem for many-to-one two-sided markets
and valid-utility games. The result for two-sided markets imply that given an instance
of the many-to-one stable matching problem, verifying if there exists a stable matching 
is NP-hard.

\begin{theorem}
\label{theo:market}
{\sc Has a singleton Sink} is {\sf NP}-hard  for (i) uniform utility-based two-sided market games,
(ii) many-to-one two-sided market games, and 
(iii) valid-utility games.
\end{theorem}

\begin{proof}
To prove {\sf NP}-completeness, we give a reduction from the {\sc 3Sat} problem.
Given an instance of the {\sc 3Sat} problem, we construct
an instance of the utility-based two-sided market game as follows:
for each variable $x_i$, we put a player $X_i$ with a one and a zero strategy. 
For each clause $c_j$, we put two players $C_j$ and $K_j$ each with a one and a zero strategy. 
We construct the game such that $C_j$ and $K_j$ have a cycle of best responses if and only if the clause is not satisfied. In other words, if the $X$-players choose a strategy profile that satisfies all clauses, all clause players eventually reach a stable solution.

The zero strategy of $C_j$ is $\{a_j,b_j\}$ and the one strategy is $\{c_j\}$.
The zero strategy of $K_j$ is $\{a_j\}$ and the one strategy is 
$\{b_j\} \cup \{r_{j,i} |\text{for all variables} x_i \text{ in clause } c_j\}$.
The $a$-markets have utility $305$ and prefer the $K$-players. The $b$-markets have utility $8$ and prefer the $C$-players. The $c$-markets have utility $310$.
The $r$-markets have utility $100$ and prefer the $X$-players.
Note that there is a best response cycle of $C_j$ and $K_j$ if and only if none of the three $r_{i,j}$-markets is allocated by an $X$-player.

The zero strategy of a player $X_i$ is $\{r_{i,j} | x_i \in c_j\} \cup \{p_{i,j} \vert \bar{x}_i \in c_j\}$. 
The one strategy of a player $X_i$ is $\{r_{i,j} | \bar{x}_i \in c_j\} \cup \{p_{i,j} | x_i \in c_j\}$. 
The $p$-markets have utility $100$. Note that both strategies have the same utility for a 
$X$-player independent of the strategy profile of other players.
Furthermore, $X_i$ gets the utility from $r_{i,j}$, if and only if it satisfies clause $c_j$, 
\end{proof}

The above theorem implies that given an instance of the many-to-one stable
matching problem, the problem of verifying if this game
has a stable matching or not is NP-hard.
Known results in the economic literature for 
many-to-one two-sided markets discuss necessary and sufficient conditions for existence of 
stable matchings (or pure Nash equilibria) for different variants of two-sided markets~\cite{GS62,RV90,KU06}, however,
before our results,  
the known results have not addressed the complexity of verifying the existence of stable matchings
(or pure Nash equilibria) given an instance of these markets.


\section{Sink Equilibria and Weighted Congestion Games}
In this section, we study the complexity of the {\sc In a Sink} and {\sc Has a Sink} problem
for weighted congestion games.
The interesting aspect of this proof is that we 
can use similar reductions for a variety of games with succinct representation. 
Applying this proof on many examples shows the strength of the proof technique.

\begin{theorem}
\label{theo:weighted}
{\sc In a Sink} is {\sf PSPACE}-hard  for weighted congestion games.
\end{theorem}
\begin{proof}

We give a reduction from the space-bounded halting problem for Turing machines. First, we reduce an instance of this problem (a TM $M$, an input $x$ and a tape bound $t$) to the halting problem for a TM $M' = (Q,\Sigma,b,\Gamma,\delta,q_0,\{q_h\})$ which simulates $M$ on $x$ without its own input. Let $\Sigma=\{0,1\}$ and $\Gamma=\{0,1,b\}$. Starting from an empty tape, $M'$ halts if and only if $M$ rejects $x$ . Furthermore, $M'$ uses additional tape cells and states for a counter that counts up to the total number of configurations of $M$. When $M$ accepts, the counter overflows, or $M$ exceeds the tape bound $t$, $M'$ erases the whole tape, moves the head to the initial position and returns to state $q_0$. $M'$ uses tape cells only right of its initial position and at most $t'$ tape cells. Note that starting from every total configuration $M'$ never stops only if $M$ rejects $x$.

To complete the proof, we construct a congestion game $G_{M'}$ that simulates Turing machine $M'$. A strategy profile $s$ which we define later is in a sink equilibrium if and only if $M'$ runs forever.
The game consists of three types of {\em configuration} players, a {\em transition player}, a set of {\em control} players, and a {\em clock} player. 
The first type of configuration players is  a {\em state} player
with $|Q|$ strategies. The second type of configuration players is a {\em position} player for the position of the head with $t'$ strategies; and the third type of configuration players is a set of {\em cell} players cell$_i$  for each tape cell  $0 \le i \le t'$ with the $|\Gamma|$ 
strategies for the content of the tape cell $i$. There is a simple bijective mapping between the strategy profiles of the configuration players and the configurations of $M'$.

The game is constructed in such a way that every sequence of improvement steps can be divided in rounds. 
At the end of a  round $i$, let $c_i$ be the configuration obtained from the strategy profile of the configuration players.
For every sequence of improvement steps, $c_1  \vdash c_2  \vdash c_3  \vdash \ldots$ denotes the run of $M'$ starting from $c_1$.


 


We now describe our construction in more details. The strategies of the configuration players are described in Figure~\ref{figure:configurationplayers}. Every strategy of a configuration player has two unique resources, an $\alpha$ resource and a $\beta$ resource. The $\alpha$ resources have delay $0$ if allocated by one player and delay $1$ otherwise. 
The $\beta$ resources have delay $0$ if allocated by one player and delay $M$ otherwise. 

\begin{figure}[h]
\begin{tabular}[ht]{lll}
state player & position player & player cell$_i$ with $0 \le i \le t'$ \\

\begin{tabular}[ht]{|l|l|l|}
\hline
  strategies   & resources & delays \\
  \hline
  $q \in Q$ & $\alpha^q $	& $0/1$\\
            & $\beta^q  $  & $0/M$\\
\hline
\end{tabular}
& 
\begin{tabular}[ht]{|l|l|l|}
\hline
  strategies   & resources & delays \\
  \hline
  $0 \le i \le t'$ & $\alpha^i $	& $0/1$\\
                          & $\beta^i  $  & $0/M$\\
\hline
\end{tabular}
& 

\begin{tabular}[ht]{|l|l|l|}
\hline
  strategies   & resources & delays \\
  \hline
  $\sigma\in \Gamma$ & $\alpha^\sigma_i $	& $0/1$\\
                     & $\beta^\sigma_i  $  & $0/M$\\
\hline
\end{tabular}
\end{tabular}

\caption{Definition of strategies of the three types of configuration players}
\label{figure:configurationplayers}
\end{figure}

\begin{figure}[ht]
\begin{tabular}{lll}

Player Control$_{W,q,i,i',\sigma}$ & Player Control$_{V,q,i,i',\sigma}$ & Control$_D$ \\
\begin{tabular}[ht]{|l|l|l|}
\hline
  Strategy   & Resources & Delays \\
  \hline
	    Zero    &  $\beta^0_{W,q,i,i',\sigma}$ & $0/M$ \\
         		  & $\alpha^0_{W,q,i,i',\sigma}$	& $0/1$\\
  \hline
	    One     &  $\beta^1_{W,q,i,i',\sigma}$ & $0/M$ \\
         		  & $\alpha^1_{W,q,i,i',\sigma}$	& $0/1$\\
\hline
\end{tabular}

&
 
\begin{tabular}[ht]{|l|l|l|}
\hline
  Strategy   & Resources & Delays \\
  \hline
	    Zero    &  $\beta^0_{V,q,i,i',\sigma}$ & $0/M$ \\
         		  & $\alpha^0_{V,q,i,i',\sigma}$	& $0/1$\\
  \hline
	    One     &  $\beta^1_{V,q,i,i',\sigma}$ & $0/M$ \\
         		  & $\alpha^1_{V,q,i,i',\sigma}$	& $0/1$\\
\hline
\end{tabular}
&

\begin{tabular}[ht]{|l|l|l|}
\hline
  Strategy   & Resources & Delays \\
  \hline
	    Zero    &  $\beta^0_{D}$ & $0/M$ \\
         		  & $\alpha^0_{D}$	& $0/1$\\
  \hline
	    One     &  $\beta^1_{D}$ & $0/M$ \\
         		  & $\alpha^1_{D}$	& $0/1$\\
\hline
\end{tabular}
\end{tabular}

\caption{Strategies of the control players, for each $q\in Q$, $0 \le i \le n$, $i' \in \{i-1,1,i+1\}$,  and $\sigma \in \Gamma$ }
\label{figure:control}
\end{figure}

Each control player has two strategies, Zero and One, which are constructed in the same manner like strategies of configuration players (see Figure~\ref{figure:control}). 
The transition player has the following strategies Wait, Done, Halt, and several strategies Read$_{q,i,\sigma}$, Write$_{q',i',i,\sigma'}$, and Verify$_{q',i',i,\sigma'}$ (for each $i,i' \in \{1,\ldots,t'\}$, $q,q' \in Q$, and $\sigma, \sigma' \in \Sigma$). The details of theses strategies and the resources they contain are listed in Figure~\ref{figure:transition}.  The clock player has two strategies, Trigger and Wait. 
Trigger contains the two resources, TriggerMain and TriggerClock. The strategy Wait contains one resource with constant delay of $110$.

Let us remark that each $\alpha$- or $\beta$-resource is allocated by at most two players; the transition player and one of the configuration or control players. The general idea is that the improvement steps for the transition player is determined by the strategy profile of the configuration and control players. That is, the transition player never deviates to a strategy that contains a $\beta$-resource which is allocated by another player. On the other hand, the transition player determines the improvement steps for configuration and control players if he allocates $\alpha$-resources. Note that each $\alpha$-resource is associated with exactly one strategy of exactly one configuration or control player. 
\begin{figure}[ht]

\begin{tabular}[ht]{|l|l|l|l|}
 \hline
  Strategy   & Resources & Delays  \\
  \hline
	Wait &  $\beta^1_{W,q',i',i,\sigma'},\beta^1_V{q',i',i,\sigma'}$ for all $q',i',i,\sigma'$  &$0/M$ \\            
	&   $\alpha^1_{D}  $ & $0/1$  \\	
	& TriggerMain & $0/100/100$  \\

	\hline
	$Read_{q,i,\sigma}$& $\beta^p$ for all $p \in Q \setminus q$  & $0/M$ \\
for each $q \in Q$,  &	$\beta^j$ for all $j \ne i$  & $0/M$  \\
$0 \le i \le t'$ and $\sigma \in \Gamma$ &   $\beta^{\sigma'}_i$ for all $\sigma' \in \Gamma \setminus \sigma$  & $0/M$  \\
	                   &   $\beta^1_{D}  $ & $0/M$\\
	                   & $\alpha^0_{W,q',i',i,\sigma'}$ with $\delta(q,\sigma)=(q',\sigma',d)$ and $i' = i +d$    & $0/1$\\
	                   & N.N.                                     & $80$\\
  \hline
	$Write_{q',i',i,\sigma'}$&
 $\alpha^p$ for all $p \in Q' \setminus q'$ & $0/1$\\
for each $q' \in Q$, $0 \le i \le t'$,  &     $\alpha^j$ for all $j \ne i'$ & $0/1$\\
$i' \in \{i-1,i,i+1\}$, &  $\alpha^{\sigma}_i$ for all $\sigma \in \Gamma \setminus \sigma'$  & $0/1$\\
 and $\sigma' \in \Gamma$  &    $\alpha^0_{V,q',i',i,\sigma'} $              & $0/1$\\
                  &  $\beta^0_{W,q',i',i,\sigma'}$ & $0/M$  \\
                  & N.N.                     	& $60$ \\
                      
   \hline
  $Verify_{q',i',i,\sigma'}$ & $\beta^p$ for all $p \in Q \setminus q'$  & $0/M$  \\
for each $q' \in Q$, $0 \le i \le t'$, &	$\beta^j$ for all $j \ne i'$  & $0/M$ \\
$i' \in \{i-1,i,i+1\}$, &   $\beta^{\sigma}_i$ for all $\sigma \in \Gamma \setminus \sigma'$  & $0/M$ \\ 
 and $\sigma' \in \Gamma$  & $\beta^0_{V,q',i',i,\sigma'}$  & $0/M$  \\
                    	& $\alpha^0_D$ & $0/1$\\
                      &	N.N. & $40$  \\
                 
	\hline
	Done  & triggerClock& $0/0/20$ \\
	      & $\beta^0_D$  &  $0/M$   \\
	     & $\alpha^1_{W,q',i',i,\sigma'},\alpha^1_{V,q',i',i,\sigma'}$ for all $q',i',i,\sigma'$  &$0/1$  \\
	     &	N.N. & $20$  \\  
	\hline
	Halt & $ \beta^q$ for all $q \in Q \setminus q_h$ & $0/M$  \\

	\hline
\end{tabular}
\caption{Definition of strategies of the transition player. 
Resources that are denoted by N.N. are used by the transition player only and have a constant delay.}
\label{figure:transition}
\end{figure}

Now, we are ready to describe the aforementioned sequence of improvement steps that corresponds to one round in more details. 
Consider any strategy profile in which the clock players are on Trigger, the transition player is on Wait and all control players except control$_D$ are on One.
Let $q$ be the strategy of the state player, $i$ the strategy of the position player and $\sigma_0,\ldots,\sigma_{t'}$ the strategies played by the players  $cell_0,\ldots,cell_{t'}$. Figure~\ref{figure:round} describes the sequence of improvement steps emerging from this strategy profile. The strategy profile at the end of the round differs from the initial one only in the choices of the configuration players. The deviations of the configuration players corresponds to a step of the Turing machine $M'$. Note that this sequence is essentially unique as there are no other improving deviations. If and only if the state player is on $q_h$, the transition player may deviate to the strategy Halt. This is a Nash equilibrium of $G_{M'}$. 
Now let $s$ be a strategy profile in which the clock players is on Trigger, the transition player on Wait, and all control players except control$_D$ on One. Let the configuration players' choice in $s$ correspond to the initial configuration of $M'$. Then, $s$ is in a sink equilibrium if and only if $M'$ does not halt.

\begin{figure}[htb]
\begin{center}
\begin{tabular}[h]{|l l|}
\hline
(1)   & The transition player deviates from Wait to Read$_{q,i,\sigma_i}$.\\
(2)   &  Player control$_{W,q',i',i,\sigma'}$ deviates to Zero.\\
(3)   &  The transition player deviates to Write$_{q',i',i,\sigma'}$.\\
(4)   & The configuration players deviate to the new configuration\\
      & and the player control$_{V,q',i',i,\sigma'}$ deviates to Zero.\\
(5)   & The transition player deviates to Verify$_{q',i',i,\sigma'}$.\\
(6)   & The player control$_D$ deviates to One.\\
(7)   & The transition player deviates to Done.\\
(8)   & The clock player deviates to Wait and\\
      & the controll players except control$_D$ deviate to Zero\\
(9)   & The transition player deviates to Wait.\\
(10)   &  The clock player deviates to Trigger and\\
       & the player control$_D$ deviates to Zero\\
\hline
\end{tabular}
\end{center}
\caption{\label{figure:round}
Description of a {\em round}.}
\end{figure}
\end{proof}

We now consider the problem {\sc Has a non-singleton Sink} for weighted congestion games.

\begin{theorem}
\label{weighted-qbf}
{\sc Has a non-singleton Sink} is {\sf PSPACE}-hard for weighted congestion games.
\end{theorem}

This results follows from the proof of Theorem~\ref{theo:weighted} and the following Lemma. The lemma implies that there is at most one unique sink equilibrium in the constructed game.

\begin{lemma}
Every Sink equilibrium contains a strategy profile in which the clock player is on Trigger, the main player on Wait and all controll players on their Zero strategy. 
\end{lemma}

\begin{proof}
If no player has delay $M$ or greater, the game converges as described in Figure~\ref{figure:round} and eventually reaches a  strategy profile in which the clock player is on Trigger, the main player on Wait and all controll players on their Zero strategy.
Note that no strategy profile with a player having delay $M$ or greater is reachable.
If players have delay of $M$ or greater, there is a sequence of improvement steps such that no player has delay of $M$ or more, e.g. each control or configuration player with delay of $M$ changes to another strategy.
\end{proof}
Thus, every sink equilibrium also contains the strategy profile that corresponds to the initial configuration of $M'$.
Therefore, there is a unique sink equilibrium if and only if $M$ rejects $x$.

\section{Sink Equilibria and Player-Specific Congestion Games}

\begin{theorem}
\label{playerspecific}
{\sc In a Sink} is {\sf PSPACE}-hard for player-specific congestion games.
\end{theorem}

One can easily replace the clock player in the construction  
which is the only player with non-uniform weight by a player
with weight $1$ and modify the (player-specific) delay functions as follows.
For the transition player the resource TriggerMain has delay $0$ if one player allocates it and delay $100$ otherwise.
For the clock player the resource TriggerMain has always delay $100$.
The delay functions of the resource TriggerClock is identical for both players. It has delay $0$ if one player allocates the resource and delay $20$ for two or more players.
For each strategy profile the delay for each player is identical to the 
delay in the previous example.

\begin{theorem}
\label{playerspecific-qbf}
{\sc Has a non-singleton Sink} is {\sf PSPACE}-hard for player-specific congestion games.
\end{theorem}
\begin{proof}
This result follows by the same argument as for Theorem~\ref{weighted-qbf}.
\end{proof}


\section{Sink Equilibria and Anonymous Games}

Next, we consider anonymous games with constant-size strategy
set and show that {\sc in a sink} for this game
is also {\sf PSPACE}-complete.

\begin{theorem}
\label{theo:anon}
{\sc In a Sink} is {\sf PSPACE}-hard for anonymous games with constant-size strategy sets.
\end{theorem}

We give a reduction from the halting problem of a space bounded Turing machine $M'$ as defined in the proof of Theorem~\ref{theo:weighted}. Additionally, we assume that states of $M'$ are denoted by $q´_0,\ldots,q_m$ where $q_m$ is the halting state. We construct an anonymous game with a constant number of strategies. Each player has a set of ({\em allowed}) strategies. Every strategy that is not allowed always has utility $0$. The only other utility values in the game are $1$ and $2$. 
Given a strategy profile $s=(s_1,\ldots,s_k)$, let $|s_i|$ denotes the number of players that play strategy $s_i$.

The game consists of the three types of {\em configuration players} and five types of {\em auxiliary players} and two {\em control players}. 
The strategy choices of the  configuration players can be mapped to configurations of the TM $M'$. Every sequence of improvement steps can be partitioned into rounds. Each round simulates one step of $M'$. 
At the end of a  round $i$, let $c_i$ be the configuration obtained from the strategy profile of the configuration players.
For every sequence of improvement steps, $c_1  \vdash c_2  \vdash c_3  \vdash \ldots$ equals the run of $M'$ starting from $c_1$.

We first describe the configuration players before we describe the remaining players and the process that simulates one step of $M'$.
The first type of configuration players are $|Q|$ identical state players that choose between the two actions state$^1$ and state$^0$. For $j = |\text{state}^1|$ corresponds to $M'$ being in state $q_j$.
The second type are $t'$ identical position players that choose between the two actions position$^1$ and position$^0$. For $p = |\text{position}^1|$ corresponds to the head of $M'$ being in  position $p$. 
The third type are the cell players cell$_0,\ldots,$cell$_{t'}$ which choose between the actions cell$^0$, cell$^1$, cell$^b$, and change. Unlike the previous two types of players, the cell players are non-identical, i.e., each player has a different payoff function.
For each $1 \le i \le t'$, player cell$_i$ on action cell$^0$ (cell$^1$ or cell$^b$) corresponds to the fact that tape cell $i$ contains the symbol $0$ ($1$ or blank).

\begin{figure}[ht]
\begin{center}
\begin{tabular}{|l|l|}
\hline
Players & allowed strategies \\
\hline
\hline
cell$_1,\ldots,\text{cell}_{t'}$   & cell$^0$,cell$^1$, cell$^b$, change \\
\hline
position$_1,\ldots,\text{position}_{t'}$& position$^1$, position$^0$ \\
\hline
state$_1,\ldots,\text{state}_m$   & state$^1$, state$^0$\\
\hline
tape$_1,\ldots,\text{tape}_{t'}$    & tape$^0$, tape$^1$, tape$^b$\\
\hline
symbol      & symbol$^0$,symbol$^1$,symbol$^b$ \\
\hline
new-sym     &  new-sym$^0$,new-sym$^1$, new-sym$^b$ \\
\hline
new-pos$_1,\ldots,\text{new-pos}_{t'}$ & new-pos$^1$, new-pos$^0$ \\
\hline
new-state$_1,\ldots,\text{new-state}_m$  & new-state$^1$, new-state$^0$\\
\hline
transition1 & init, tape-change, eval-tape, new-sym, new-sym2, new-pos,\\
            & new-pos2, new-state, new-state2, halt\\
\hline
transition2 & Xinit, Xtape-change, Xeval-tape, Xfnew-sym, Xnew-sym2, \\
			      &  Xnew-pos Xnew-pos2, Xnew-state, Xnew-state2  \\
\hline
\end{tabular}
\end{center}
\caption{Players and their strategies
\label{fig:players-anon}
}
\end{figure}

There are five types of {\em auxiliary players} and two {\em control players}. All players and their allowed strategies are listed in Figure~\ref{fig:players-anon}. The utility functions for each player are described in Appendix~\ref{appendix-anon}.
The players $tape_1,\ldots,tape_{t'}$ have identical payoff functions. They are used to evaluate symbol at the current position. The player symbol saves this symbol. The players new-sym, new-pos$_1,\ldots,\text{new-pos}_{t'}$, new-state$_1,\ldots,\text{new-state}_m$ calculate the changes to the configuration. 
The control players ensure that strategy changes happen in a certain order that corresponds to one step.


\begin{lemma} 
Let $c$ be a configuration of $M'$ and $c'$ the successor configuration.
Every sequence of improvement steps from a strategy profile in which the configuration players play corresponding to $c$ and the first control player is on init, reaches a strategy profile in which the configuration players play corresponding to $c'$ and the first control player is on init.
\end{lemma}

\begin{proof}
We now describe this sequence of improvement steps which we call a round. It is listed in Figure~\ref{fig:step-anon} in detail. One can easily check for each of the strategy profiles that the next one is essentially unique.

In a round, the first control player successively changes through his strategies (c.f. steps (2),(4),...). The second control player follows his choices in his corresponding strategies. By construction of the payoff function, this ensures that the control players only change their strategies in a certain order.
Each of these steps of the first control player is interrupted by improvement steps of subsets of configuration or auxiliary players.
The utility functions (cf. Figure~\ref{fig:util-anon}) are designed in such a way that these improvement steps are possible if and only if the control player plays the corresponding strategy. Additionally, the control player may only continue with his next step after these other player have changed their strategies (cf. Figure~\ref{fig:util-anon-control1}) . 

We now describe the improvement steps of the configuration and auxiliary players only.
Consider any strategy profile of the configuration players and assume the first control player is on init (strategy profile (1) in Figure~\ref{fig:step-anon} in Appendix~\ref{appendix-anon}.
The $t'$ tape players change to a strategy profile in that the number of players on tape$^0$, tape$^1$, and tape$^b$ equals the number of players on cell$^0$, cell$^1$, and cell$^b$ (2). 
The player cell$_i$ with $i = |$position$^i|$ changes to his strategy to change (4).
The symbol player changes to symbol$^0$, symbol$^1$, or symbol$^b$ depending on which strategy was left by the player cell$_i$ (6). This can be coded into the utility function by evaluating the difference of number of players in the cell and tape strategies.
The player new-symbol changes to the strategy new-symbol$^{\sigma'}$ where $\sigma'$ corresponds to the new symbol (8). This can be coded as a function as from number of players on symbol$^0$, symbol$^1$, symbol$^b$, and state$^1$. 
The player cell$_i$ changes to the strategy cell$^{\sigma'}$ (10).
Exactly $i'$ players choose new-pos$^1$ where $i'$ is the new position of $M'$ (12).
The players position change their strategies such that $|$position$^1| = |$new-pos$^1| = i'$ (14).
Exactly $q'$ players new-state choose new-state$^1$ where $q_{q'}$ is the new state of $M'$ (16).
The players state change their strategies such that $|$state$^1| = |$new-state$^1| = q'$ (18).
The configuration players' strategy profile now corresponds to the new configuration after one step of $M'$. 
\end{proof}

\begin{theorem}
\label{weighted-exist}
{\sc Has a non-singleton Sink} is {\sf PSPACE}-hard for anonymous games.
\end{theorem}

\begin{proof}
By construction of $M'$ and the proof of Theorem~\ref{theo:anon}, it suffices to show that every infinite sequence of improvement steps contains a strategy profile with player control1 on init, i.e. a profile listed in the first row of Table~\ref{fig:step-anon}.

The strategy changes of control1 have to occur in the same order as listed in Table~\ref{fig:step-anon}. Therefore, every sequence with infinite strategy changes of control1 contains a profile with control1 on init. We, therefore, show that there is no
infinite sequence that contains no strategy change of control1. 
Thus, fix any strategy choice for player control1. Observe that the utility functions of the remaining players (cf. Figure~\ref{fig:util-anon}) do not allow an infinite sequence. 
\end{proof}
\section{Sink Equilibria in other Games}

\begin{theorem}
\label{theo:market-in}
{\sc In a Sink} is {\sf PSPACE}-hard for (i) 
uniform utility-based two-sided market games,
(ii) many-to-one two-sided market games, and 
(iii) valid-utility games.
\end{theorem}

\begin{theorem}
\label{markets-exist}
{\sc Has a non-singleton Sink} is {\sf PSPACE}-hard for (i) uniform utility-based two-sided market games,
(ii) many-to-one two-sided market games, and 
(iii) valid-utility games.
\end{theorem}

\noindent The proof is a rework of the proof for Theorem~\ref{theo:weighted}
and is shifted to Appendix~\ref{app:market}.
The Nash dynamics of the uniform utility-based two-sided market game that we describe there is isomorphic to the Nash dynamics of the congestion game in the proof for Theorem~\ref{theo:weighted}.


\pagebreak
\begin{appendix}

\section{Proof of Theorem~\ref{theo:market-in}}
\label{app:market}


The Nash dynamics of the uniform utility-based two-sided market game that we describe here is isomorphic to the Nash dynamics of the congestion game in the proof for Theorem~\ref{theo:weighted}. Thus, all properties easily transfer. The strategies of the transition player and the preferences of the markets can be found in Figure~\ref{transition_markets}. The strategies of the remaining players can be obtained from the previous proof.

\begin{figure}[ht]

\begin{tabular}[ht]{|l|l|l|l|}
 \hline
  Strategy   & Markets & Utilities (Preference) \\
  \hline
	Wait &  $\beta^1_{W,q',i',i,\sigma'}$ for all $q',i',i,\sigma'$  &$ M$ (Control$_{W,q',i',i,\sigma'}$, transition player) \\

	    & $ \beta^1_V{q',i',i,\sigma'}$ for all $q',i',i,\sigma'$  &$ M$ (Control$_{V,q',i',i,\sigma'}$, transition player) \\ 
	           
	&   $\alpha^1_{D}  $ & $ 1$ (transition player, Control$_D$) \\	
	& TriggerMain & $ 100$ (clock player, transition player) \\

	\hline
	$Read_{q,i,\sigma}$& $\beta^p$ for all $p \in Q \setminus q$  & $ M$ (state player, transition player)\\
for each $q \in Q$,  &	$\beta^j$ for all $j \ne i$  & $ M$ (position player, transition player)  \\
$0 \le i \le t'$ and $\sigma \in \Gamma$ &   $\beta^{\sigma'}_i$ for all $\sigma' \in \Gamma \setminus \sigma$  & $ M$  (cell$_i$, transition player) \\
	                   &   $\beta^1_{D}  $ & $ M$ (Control$_D$, transition player)\\
	                   & $\alpha^0_{W,q',i',i,\sigma'}$ with $\delta(q,\sigma)=(q',\sigma',d)$    & $ 1$ (transition player,Control$_{W,q',i',i,\sigma'}$) \\
	                   & and $i' = i +d$ &\\
	                   & N.N.                                     & $N-(|Q| +t' + | \Gamma| -1)M + 20 $\\
  \hline
	$Write_{q',i',i,\sigma'}$&
 $\alpha^p$ for all $p \in Q' \setminus q'$ & $1$ (transition player, state player)\\
for each $q' \in Q$, $0 \le i \le t'$,  &     $\alpha^j$ for all $j \ne i'$ & $ 1$ (transition player, position player)\\
$i' \in \{i-1,i,i+1\}$, &  $\alpha^{\sigma}_i$ for all $\sigma \in \Gamma \setminus \sigma'$  & $ 1$ (transition player, cell$_i$\\
 and $\sigma' \in \Gamma$  &    $\alpha^0_{V,q',i',i,\sigma'} $              & $ 1$ (transition player, Control$_{V,q',i',i,\sigma'}$)\\
                  &  $\beta^0_{W,q',i',i,\sigma'}$ & $ M$ (Control$_{W,q',i',i,\sigma'}$,transition player) \\
                  & N.N.                     	& $N-M+40$ \\
                      
   \hline
  $Verify_{q',i',i,\sigma'}$ & $\beta^p$ for all $p \in Q \setminus q'$  & $ M$ (state player, transition player)  \\
for each $q' \in Q$, $0 \le i \le t'$, &	$\beta^j$ for all $j \ne i'$  & $ M$ (position player, transition player)\\
$i' \in \{i-1,i,i+1\}$, &   $\beta^{\sigma}_i$ for all $\sigma \in \Gamma \setminus \sigma'$  & $ M$ (cell$_i$, transition player) \\ 
 and $\sigma' \in \Gamma$  & $\beta^0_{V,q',i',i,\sigma'}$  & $ M$ (Control$_{V,q',i',i,\sigma'}$,transition player)  \\
                    	& $\alpha^0_D$ & $ 1$ (transition player, Control$_D$)\\
                      &	N.N. & $N-(|Q|+t' + |\Gamma|-1)M + 60$  \\
                 
	\hline
	Done  & triggerClock& $  80$ (transition player, clock player)\\
	      & $\beta^0_D$  &  $ M$(Control$_D$, transition player)   \\
	     & $\alpha^1_{W,q',i',i,\sigma'}$ for all $q',i',i,\sigma'$  &$ 1$ (transition player, Control$_{W,q',i',i,\sigma'}$)  \\
	  & $\alpha^1_{V,q',i',i,\sigma'}$ for all $q',i',i,\sigma'$  &$ 1$ (transition player, Control$_{V,q',i',i,\sigma'}$)  \\   
	     &	N.N. & $N-M+20$  \\  
	\hline
	Halt & $ \beta^q$ for all $q \in Q \setminus q_h$ & $ M$ (state player, transition player)  \\
	     & N.N.                                      & $N-M$ \\

	\hline
\end{tabular}
\caption{Strategies of the transition players. 
Markets denoted by N.N. are used by the transition players only.
Let $N = |Q|(t+1)6|\Gamma|M$}.
\label{transition_markets}
\end{figure}

\clearpage

\section{Details of the proof of Theorem~\ref{theo:anon}}
\label{appendix-anon}

\begin{sidewaystable}[ht]
\begin{tabular}{|r|l|c|c|c|c|c|c|c|}
\hline
 & Configuration players & tape & symbol & new-sym & new-pos & new-state & control1 & control2 \\
\hline

1 & $(\sigma_1,\ldots,\sigma_{i-1},\sigma_{i},\sigma_{i+1}\ldots \sigma_{t'}),q,i$& \underline{*} & * & * & * & * & init & \underline{*} \\

2 &$(\sigma_1,\ldots,\sigma_{i-1},\sigma_{i},\sigma_{i+1}\ldots \sigma_{t'}),q,i$&$p_0,p_1,p_b$ & * & * & * & * & \underline{init} & Xinit \\

3 &$(\sigma_1,\ldots,\sigma_{i-1},\underline{\sigma_{i}},\sigma_{i+1}\ldots \sigma_{t'}),q,i$&$p_0,p_1,p_b$ & * & * & * & * & tape-change & \underline{Xinit} \\

4 &$(\sigma_1,\ldots,\sigma_{i-1},\text{change},\sigma_{i+1}\ldots \sigma_{t'}),q,i$&$p_0,p_1,p_b$ & $*$ & * & * & * & \underline{tape-change} & Xtape-change \\

5 &$(\sigma_1,\ldots,\sigma_{i-1},\text{change},\sigma_{i+1}\ldots \sigma_{t'}),q,i$&$p_0,p_1,p_b$ & $\underline{*}$ & * & * & * &eval-tape &  \underline{Xtape-change} \\

6 &$(\sigma_1,\ldots,\sigma_{i-1},\text{change},\sigma_{i+1}\ldots \sigma_{t'}),q,i$&$p_0,p_1,p_b$ & $\sigma_i$ & * & * & * & \underline{eval-tape} & Xeval-tape \\

7 &$(\sigma_1,\ldots,\sigma_{i-1},\text{change},\sigma_{i+1}\ldots \sigma_{t'}),q,i$&$p_0,p_1,p_b$ & $\sigma_i$ & \underline{*} & * & * & new-sym &  \underline{Xeval-tape}\\

8 &$(\sigma_1,\ldots,\sigma_{i-1},\text{change},\sigma_{i+1}\ldots \sigma_{t'}),q,i$&$p_0,p_1,p_b$ & $\sigma_i$ & $\sigma'$ & * & * & \underline{new-sym} & Xnew-sym\\

9 &$(\sigma_1,\ldots,\sigma_{i-1},\underline{\text{change}},\sigma_{i+1}\ldots \sigma_{t'}),q,i$&$p_0,p_1,p_b$ & $\sigma_i$ & $\sigma'$ & * & * & new-sym2 &  \underline{Xnew-sym}\\

10 &$(\sigma_1,\ldots,\sigma_{i-1},\sigma',\sigma_{i+1}\ldots \sigma_{t'}),q,i$&$p_0,p_1,p_b$ & $\sigma_i$ & $\sigma'$ & * & * & \underline{new-sym2} & Xnew-sym2\\

11 &$(\sigma_1,\ldots,\sigma_{i-1},\sigma',\sigma_{i+1}\ldots \sigma_{t'}),q,i$&$p_0,p_1,p_b$ & $\sigma_i$ & $\sigma'$ & \underline{*} & * & new-pos 
& \underline{Xnew-sym}\\

12 &$(\sigma_1,\ldots,\sigma_{i-1},\sigma',\sigma_{i+1}\ldots \sigma_{t'}),q,i$&$p_0,p_1,p_b$ & $\sigma_i$ & $\sigma'$ & $i'$ & * & \underline{new-pos} & Xnew-pos\\

13 &$(\sigma_1,\ldots,\sigma_{i-1},\sigma',\sigma_{i+1}\ldots \sigma_{t'}),q,\underline{i}$&$p_0,p_1,p_b$ & $\sigma_i$ & $\sigma'$ & $i'$ & * & new-pos2 &  \underline{Xnew-pos}\\

14 &$(\sigma_1,\ldots,\sigma_{i-1},\sigma',\sigma_{i+1}\ldots \sigma_{t'}),q,i'$&$p_0,p_1,p_b$ & $\sigma_i$ & $\sigma'$ & $i'$ & * & \underline{new-pos2} & Xnew-pos2\\

15 &$(\sigma_1,\ldots,\sigma_{i-1},\sigma',\sigma_{i+1}\ldots \sigma_{t'}),q,i'$&$p_0,p_1,p_b$ & $\sigma_i$ & $\sigma'$ & $i'$ & \underline{*} & new-state& \underline{Xnew-pos2}\\

16 &$(\sigma_1,\ldots,\sigma_{i-1},\sigma',\sigma_{i+1}\ldots \sigma_{t'}),q,i'$&$p_0,p_1,p_b$ & $\sigma_i$ & $\sigma'$ & $i'$ & $q'$ & \underline{new-state} & Xnew-state\\

17 &$(\sigma_1,\ldots,\sigma_{i-1},\sigma',\sigma_{i+1}\ldots \sigma_{t'}),\underline{q},i'$&$p_0,p_1,p_b$ & $\sigma_i$ & $\sigma'$ & $i'$ & $q'$ & new-state2 &  \underline{Xnew-state}\\

18 &$(\sigma_1,\ldots,\sigma_{i-1},\sigma',\sigma_{i+1}\ldots \sigma_{t'}),q',i'$&$p_0,p_1,p_b$ & $\sigma_i$ & $\sigma'$ & $i'$ & $q'$ & \underline{new-state2} & Xnew-state2\\

19 &$(\sigma_1,\ldots,\sigma_{i-1},\sigma',\sigma_{i+1}\ldots \sigma_{t'}),q',i'$&$\underline{p_0,p_1,p_b}$ & $\sigma_i$ & $\sigma'$ & $i'$ & $q'$ & init & \underline{Xnew-state2} \\
\hline
\end{tabular}
\caption{This figure shows the sequence of strategy profiles during one round. A strategy profile is described es follows. The strategy profile of the cell players is given as a vector $\sigma \in \{0,1,b,change\}^{t'}$ where $\sigma_i$ denotes strategy cell$^{\sigma_i}$  for player cell$_i$. For the state, position, new-pos, new-state players, we give the number of players on state$^1$, position$^1$, new-pos$^1$, and new-state$^1$, respectively. The strategy profile of the tape players  is described by a vector $p \in \{0,\ldots,t'\}^3$ that denotes the number of players on tape$^0$,tape$^1$, and tape$^b$, respectively. For the players symbol and new-sym, $\sigma$ denotes strategy symbol$^\sigma$ and new-sym$^\sigma$, respectively. 
The round starts with each player  cell$_i$ on $\sigma_i \in \{0,1,b\}$ , $q$ state players on state$^1$, $i$ position players on position$^1$ and the first control player on init. 
The underlined strategies indicate the players that have an incentive to deviate from their current strategies.
\label{fig:step-anon}
}
\end{sidewaystable}

\begin{figure}[h!]
\begin{tabular}{|l|l|l|}
\hline
Player &  strategy & partitions with utility $2$\\
\hline
cell$_i$ &change   &    $|$tape-change$| \ne 0$ and $|$position$^1| = i$ \\
         &cell$^0$ &    $|$new-tape$| \ne 0$ and $|$new-sym$^0| > 0$\\
         &cell$^1$ &    $|$new-tape$| \ne 0$ and $|$new-sym$^1| > 0$\\
         &cell$^b$ &    $|$new-tape$| \ne 0$ and $|$new-sym$^b| > 0$\\
\hline
tape$_i$& tape$^0$ & $|$init$| \ne 0$ and $|$cell$^0| > |$tape$^0|$\\ 
        & tape$^1$ & $|$init$| \ne 0$ and $|$cell$^1| > |$tape$^1|$\\
        & tape$^b$ & $|$init$| \ne 0$ and $|$cell$^b| > |$tape$^b|$\\ 
\hline
symbol        & symbol$^0$  & $|$eval-tape$| \ne 0$ and $|$cell$^0|$-$|$tape$^0| < 0$ \\
       			  & symbol$^1$  & $|$eval-tape$| \ne 0$ and $|$cell$^1|$-$|$tape$^1| < 0$ \\
              & symbol$^b$  & $|$eval-tape$| \ne 0$ and $|$cell$^b|$-$|$tape$^b| < 0$ \\
\hline
new-sym       &  new-sym$^0$       & $|$new-symbol$| \ne 0$ and if $0$ is {\em new symbol} \\
              &  new-sym$^1$       & $|$new-symbol$| \ne 0$ and if $1$ is {\em new symbol} \\
              &  new-sym$^b$       & $|$new-symbol$| \ne 0$ and if $b$ is {\em new symbol} \\
\hline
new-pos      & new-pos$^1$ & $|$new-pos$| \ne 0$ and $|$new-pos$^1| < $ {\em new position}\\
             & new-pos$^0$ & $|$new-pos$| \ne 0$ and $|$new-pos$^1| > $ {\em new position}\\
\hline
new-state    & new-state$^1$ & $|$new-state$| \ne 0$ and $|$new-state$^1| >$ {\em new state}\\
             & new-state$^0$ & $|$new-state$| \ne 0$ and $|$new-state$^1| <$ {\em new state}\\
\hline
position     & position$^1$ & $|$new-pos2$| \ne 0$ and $|$position$^1| < |$new-pos$^1|$\\
             & position$^0$ & $|$new-pos2$| \ne 0$ and $|$position$^1| > |$new-pos$^1|$\\
\hline  
state        & state$^1$ & $|$new-state2$| \ne 0$ and $|$state$^1| < |$new-state$^1|$\\
             & state$^0$ & $|$new-state2$| \ne 0$ and $|$state$^1| > |$new-state$^1|$\\
\hline  
halt	       & $|$state$^1| = m$\\
\hline
\end{tabular}
\caption{The strategy partition combinations are listed that induce utility $2$.  Note that the {\em new symbol, new position}, and {\em new state} can be coded as a function of $|$symbol$^0|$,$|$symbol$^1|$,$|$symbol$^b|$, and $|$state$^1|$. 
\label{fig:util-anon}}
\end{figure}

\begin{figure}[h!]
\begin{tabular}{|l|l|l|}
\hline
 strategy & partitions with utility $2$\\
\hline
 tape-change & $|$Xinit$| > 0$ and  $|$cell$^0| = |$tape$^0|$ and $|$cell$^1| = |$tape$^1|$\\
             &  and  $|$cell$^b| = |$tape$^b|$\\
\hline
eval-tape    & $|$Xtape-change$| > 0$ and $|$cell-change$| = 1$\\ 
\hline
new-sym    & $|$Xeval-tape$| > 0$ and  $|$cell$^0|+|$symbol$^0| = |$tape$^0|$ and $|$cell$^1|+|$symbol$^1| = |$tape$^1|$ \\
 						 & and  $|$cell$^b|+|$symbol$^b| = |$tape$^b|$\\
\hline
new-sym2   & $|$Xnew-sym$| > 0$ and $|$new-sym$^{\sigma'}| = 1$ for $\sigma' =$ {\em new symbol} \\
\hline
new-pos      & $|$Xnew-sym2$| > 0$ and $|$change$| = 0$\\
\hline

new-pos2     & $|$Xnew-pos$| > 0$ and $|$new-pos$^1| = $ {\em new position}\\
\hline

new-state    & $|$Xnew-pos2$| > 0$ and d $|$position$^1| = |$new-pos$^1|$\\
\hline

new-state2   & $|$Xnew-state$| > 0$ and  $|$new-state$^1| = $ {\em new state}\\
\hline

init         & $|$Xnew-state2$| > 0$ and $|$state$^1| = |$new-state$^1|$\\
\hline
stop         & $|$position$^1| = m $ \\
\hline
\end{tabular}
\caption{The strategy/partition combinations of the first control player are listed that induce utility of $2$. 
\label{fig:util-anon-control1}}
\end{figure}

\begin{figure}[h!]
\begin{tabular}{|l|l|l|}
\hline
 strategy & partitions with utility $2$\\
\hline
\hline

Xinit         & $|$init$| > 0$\\
\hline
Xtape-change         & $|$tape-change$| > 0$\\
\hline
Xeval-tape         & $|$eval-tape$| > 0$\\
\hline
Xnew-sym         & $|$new-sym$| > 0$\\
\hline
Xnew-sym2         & $|$new-sym2$| > 0$\\
\hline
Xnew-pos         & $|$new-pos$| > 0$\\
\hline
Xnew-pos2         & $|$new-pos2$| > 0$\\
\hline
Xnew-state         & $|$new-state$| > 0$\\
\hline
Xnew-state2         & $|$new-state2$| > 0$\\
\hline
\end{tabular}
\caption{The strategy/partition combinations of the second control player that induce utility of $2$. 
\label{fig:util-anon-control2}}
\end{figure}

\end{appendix}


\begin{thebibliography}{10}

\bibitem{ARV06a}
H.~Ackermann, H.~R\"oglin, and B.~V\"ocking.
\newblock Pure {N}ash equilibria in player-specific and weighted congestion
  games.
\newblock In {\em Proceedings of the 2nd International Workshop on Internet and
  Network Economics (WINE)}, pages 50--61, 2006.

\bibitem{AS07}
Heiner Ackermann and Alexander Skopalik.
\newblock On the complexity of pure {N}ash equilibria in player-specific
  network congestion games.
\newblock In {\em In Proceedings of 3nd International Workshop on Internet and
  Network Economics (WINE)}, pages 419--430, 2007.

\bibitem{BFH07}
Felix Brandt, Felix~A. Fischer, and Markus Holzer.
\newblock Symmetries and the complexity of pure nash equilibrium.
\newblock In Wolfgang Thomas and Pascal Weil, editors, {\em STACS}, volume 4393
  of {\em Lecture Notes in Computer Science}, pages 212--223. Springer, 2007.

\bibitem{CD06}
Xi~Chen and Xiaotie Deng.
\newblock Settling the complexity of two-player nash equilibrium.
\newblock In {\em FOCS}, pages 261--272, 2006.

\bibitem{CDT06}
Xi~Chen, Xiaotie Deng, and Shang-Hua Teng.
\newblock Computing nash equilibria: Approximation and smoothed complexity.
\newblock In {\em FOCS}, pages 603--612, 2006.

\bibitem{DP07}
C.~Daskalakis and C.~H. Papadimitriou.
\newblock Computing equilibria in anonymous games.
\newblock In {\em {IEEE} Symposium on Foundations of Computer Science}, 2007.

\bibitem{DP07b}
C.~Daskalakis and C.~H. Papadimitriou.
\newblock On the exhaustive algorithm for nash equilibria.
\newblock page Unpublished Manuscript, 2007.

\bibitem{DGP06}
Constantinos Daskalakis, Paul~W. Goldberg, and Christos~H. Papadimitriou.
\newblock The complexity of computing a nash equilibrium.
\newblock In {\em STOC}, pages 71--78, 2006.

\bibitem{DS06}
Juliane Dunkel and Andreas~S. Schulz.
\newblock On the complexity of pure-strategy {N}ash equilibria in congestion
  and local-effect games.
\newblock In {\em In Proceedings of 2nd International Workshop on Internet and
  Network Economics (WINE)}, pages 62--73, 2006.

\bibitem{FPT04}
A.~Fabrikant, C.~Papadimitriou, and K.~Talwar.
\newblock On the complexity of pure equilibria.
\newblock In {\em Proceedings of the 36th Annual ACM Symposium on Theory of
  Computing (STOC)}, pages 604--612, 2004.

\bibitem{FP08}
Alex Fabrikant and Christos~H. Papadimitriou.
\newblock The complexity of game dynamics: Bgp oscillations, sink equilibria,
  and beyond.
\newblock In {\em SODA '08: Proceedings of the nineteenth annual ACM-SIAM
  symposium on Discrete algorithms}, pages 844--853, Philadelphia, PA, USA,
  2008. Society for Industrial and Applied Mathematics.

\bibitem{FGMS06}
L.~Fleischer, M.~Goemans, V.~S. Mirrokni, and M.~Sviridenko.
\newblock Tight approximation algorithms for maximum general assignment
  problems.
\newblock In {\em Proceedings of the 16th Annual ACM--SIAM Symposium on
  Discrete Algorithms (SODA)}, pages 611--620, 2006.

\bibitem{GS62}
D.~Gale and L.~Shapley.
\newblock College admissions and the stability of marriage.
\newblock {\em American Mathematical Monthly}, 69:9--15, 1962.

\bibitem{GLMT04}
M.~Goemans, L.~Li, V.~S. Mirrokni, and M.~Thottan.
\newblock Market sharing games applied to content distribution in ad-hoc
  networks.
\newblock In {\em Proceedings of the 5th ACM International Symposium on Mobile
  Ad Hoc Networking and Computing (MobiHoc)}, pages 1020--1033, 2004.

\bibitem{GMV05}
M.~Goemans, V.~S. Mirrokni, and A.~Vetta.
\newblock Sink equilibria and convergence.
\newblock In {\em FOCS}, 2005.

\bibitem{JL07}
Albert~Xin Jiang and Kevin Leyton-Brown.
\newblock Computing pure nash equilibria in symmetric {Action-Graph Games}.
\newblock In {\em Association for the Advancement of Artificial Intelligence
  (AAAI)}, pages 79--85, 2007.

\bibitem{KU06}
F.~Kojima and {\"{U}}.~Unver.
\newblock Random paths to pairwise stability in many-to-many matching problems:
  a study on market equilibration.
\newblock {\em International Journal of Game Theory}, 2006.

\bibitem{M96}
I.~Milchtaich.
\newblock Congestion games with player-specific payoff functions.
\newblock {\em Games and Economics Behavior}, 13:111--124, 1996.

\bibitem{R73}
R.~W. Rosenthal.
\newblock {A class of games possessing pure-strategy {N}ash equilibria}.
\newblock {\em International Journal of Game Theory}, 2:65--67, 1973.

\bibitem{RV90}
A.~E. Roth and J.~H.~Vande Vate.
\newblock Random paths to stability in two-sided matching.
\newblock {\em Econometrica}, 58(6):1475--1480, 1990.

\bibitem{SV07}
Alexander Skopalik and Berthold V\"{o}cking.
\newblock Inapproximability of pure {N}ash equilibria.
\newblock In {\em STOC '08: Proceedings of the 40th annual ACM symposium on
  Theory of computing}, pages 355--364, New York, NY, USA, 2008. ACM.

\bibitem{V02}
A.~Vetta.
\newblock {N}ash equilibria in competitive societies, with applications to
  facility location, traffic routing and auctions.
\newblock In {\em 43rd Symp. on Foundations of Computer Science (FOCS)}, pages
  416--425, 2002.

\end{thebibliography}
\end{document}